\documentclass[twocolumn]{autart}    
   
\usepackage[cmex10]{amsmath}
\usepackage{epsfig,amsfonts,amssymb,cite,times,subfigure,floatflt,wrapfig,enumerate,tikz,url,bm}
\usepackage{algorithmicx,algpseudocode}
\usepackage{booktabs}
\usepackage{siunitx}
\usepackage{graphicx}          

\newtheorem{example}{Example}

\newtheorem{lemma}{Lemma}

\newtheorem{definition}{Definition}[section]
\newtheorem{proof}{Proof}

\newcommand{\mat}[1]{\bm{#1}}
\newcommand{\ten}[1]{\bm{\mathcal{#1}}}

\newcommand{\kpr}[1]{\textsuperscript{\textcircled{#1}}}

\begin{document}

\begin{frontmatter}

\title{A Tensor Network Kalman filter with an application in recursive MIMO Volterra system identification}


\author[HKU]{Kim Batselier}\ead{kim.batselier@eee.hku.hk},    
\author[HKU]{Zhongming Chen}\ead{zmchen@eee.hku.hk},    
\author[HKU]{Ngai Wong}\ead{nwong@eee.hku.hk},               

\address[HKU]{The Department of Electrical and Electronic Engineering, The University of Hong Kong}  


\begin{abstract}                          
This article introduces a Tensor Network Kalman filter, which can estimate state vectors that are exponentially large without ever having to explicitly construct them. The Tensor Network Kalman filter also easily accommodates the case where several different state vectors need to be estimated simultaneously. The key lies in rewriting the standard Kalman equations as tensor equations and then implementing them using Tensor Networks, which effectively transforms the exponential storage cost and computational complexity into a linear one. We showcase the power of the proposed framework through an application in recursive nonlinear system identification of high-order discrete-time multiple-input multiple-output (MIMO) Volterra systems. The identification problem is transformed into a linear state estimation problem wherein the state vector contains all Volterra kernel coefficients and is estimated using the Tensor Network Kalman filter. The accuracy and robustness of the scheme are demonstrated via numerical experiments, which show that updating the Kalman filter estimate of a state vector of length $10^9$ and its covariance matrix takes about 0.007s on a standard desktop computer in Matlab.
\end{abstract}

\end{frontmatter}

\section{Introduction}
After its publication in 1960, the Kalman filter~\cite{kalman1960new} was quickly adopted into the Apollo onboard guidance system~\cite{mcgee1985discovery} and has found many other applications ever since. The square-root filter is a more numerically stable implementation and was first developed by Potter~\cite{potter1963statistical}. It replaces the covariance matrix in the Kalman filter equations by its Cholesky factor, which is better conditioned. Over the next decade other numerically stable implementations, which also use Cholesky factors, were developed~\cite{bierman1977factorization,carlson1973fast,morf1975square}. Extensions of the Kalman filter to nonlinear models are the Extended Kalman filter (EKF)~\cite{schmidt1976practical,jazwinski2007stochastic}, the statistically linearized filter (SLF)~\cite{gelb1974applied} and the unscented Kalman filter (UKF)~\cite{julier1995new,julier2004unscented}. All these filters turn out to be specific instances of Bayesian filters~\cite{bayesianfiltering}, where the Kalman filter solution emerges from the assumption that both the dynamic and measurement models are linear Gaussian.

The Kalman filter is inherently limited by the length of the state vector that is to be estimated. For example, using a Kalman filter to estimate a state vector with a length $n^d$ will quickly become intractable, even for moderate sizes of $n$ and $d$. In this article, we explain how Tensor Networks\cite{TNorus,Cichocki2014} enable the estimation of exponentially long state vectors in a computationally efficient manner. The main paradigm used in the Tensor Network framework is to represent the exponentially long state vectors and their corresponding covariance matrices as tensors in a network. These tensors are called the Tensor Network cores and all computations of the Kalman filter are performed directly on the cores. We show in Section \ref{sec:implementation} that this reduces the computational complexity and storage cost from $O(n^d)$ to $O(dn)$.

A particularly well-suited application of the Tensor Network Kalman filter is the recursive identification of discrete-time multiple-input-multiple-output (MIMO) Volterra systems~\cite{wiener2013nonlinear,R:81}. These nonlinear systems have been extensively studied and applied in applications like speech modeling~\cite{Mumolo1993}, loudspeaker linearization~\cite{Kajikawa2008}, nonlinear control~\cite{fj2012identification}, active noise control~\cite{Tan2001}, modeling of biological and physiological systems~\cite{korenberg1996identification}, nonlinear communication channel identification and equalization~\cite{cheng2001optimal,fernandesmotalnlm82}, distortion analysis~\cite{Wambacq1998} and many others. Their applicability has been limited however to ``weakly nonlinear systems", where the nonlinear effects play a non-negligible role but are dominated by the linear terms. This limitation is not inherent to the Volterra series themselves, as they can also represent strongly nonlinear dynamical systems, but is due to the exponentially growing number of Volterra kernel coefficients as the degree increases. Indeed, assuming a finite memory $M$, the $d$th-order response of a discrete-time single-input single-output (SISO) Volterra system is given by
\begin{align*}
y_d(t) &= \sum_{k_1,\ldots,k_d=0}^{M-1} h_d(k_1,\ldots,k_d)\,\prod_{i=1}^{d} u(t-k_i), 
\end{align*}
where $y_d(t),u(t)$ are the scalar output and input at time $t$ respectively and the $d$th-order Volterra kernel $h_d(k_1,\ldots,k_d)$ is described by $M^d$ numbers. For a multiple-input multiple-output (MIMO) Volterra system with $p$ inputs the situation gets even worse, since the $d$th-order Volterra kernel for one particular output is characterized by $(pM)^d$ numbers. This exponential growth of the number of kernel coefficients is one particular example of the infamous \emph{curse of dimensionality}.

In order to apply the Tensor Network Kalman filter to the problem of recursive system identification of MIMO Volterra systems, we first rewrite the MIMO Volterra system as a linear state space model of the Volterra kernel coefficients. The system identification problem is in this way converted into a state estimation problem. The linear state space description of SISO Volterra systems for the identification of its kernel coefficients has appeared in \cite{Barner2006}. The curse of dimensionality however limits the application of their method to low degree Volterra systems. After having converted the MIMO Volterra system into a linear state space mode, we present a Tensor Network description of MIMO Volterra systems~\cite{MVMALS}. This description effectively enables the use of a Tensor Network Kalman filter to solve the state estimation problem. In contrast with the system identification method described in~\cite{MVMALS}, the Kalman filter approach explicitly takes the effect of measurements noise into account. Furthermore, we derive how the Tensor Network cores are initialized without the explicit construction of the prohibitively large mean vectors and covariance matrices. In short, the main contributions of this article are
\begin{itemize}
\item the Kalman filter equations are rewritten as tensor equations to accommodate for the estimation of multiple state vectors at once,
\item each of the Kalman filter tensor equations are computed in the Tensor Network format, resulting in a significant reduction of computational complexity and storage cost,
\item the Tensor Network Kalman filter is applied for the recursive identification of MIMO Volterra systems.
\end{itemize}

The outline of this article is as follows. In Section \ref{sec:prelim} we give a brief overview of important tensor concepts and Tensor Network theory.  The Tensor Network Kalman filter is derived in Section \ref{sec:kalman} and its implementation is discussed in Section \ref{sec:implementation}. The MIMO Volterra Tensor Network framework from \cite{MVMALS} is reviewed and the application of the Tensor Network Kalman filter to the system identification problem is discussed in Section \ref{sec:VolterraTensor}. In Section \ref{sec:experiments}, numerical experiments demonstrate the accuracy and computational efficiency of the Tensor Network Kalman filter when applied for recursive MIMO Volterra system identification. Matlab/Octave implementations of our algorithms are freely available from \url{https://github.com/kbatseli/TNKalman}.

\section{Preliminaries}
\label{sec:prelim}
\subsection{Tensor basics}
Tensors in this article are multi-dimensional arrays that generalize the notions of vectors and matrices to higher orders. A $d$-way or $d$th-order tensor is denoted $\ten{A} \in \mathbb{R}^{n_1 \times  n_2 \times \cdots \times n_d}$ and hence each of its entries $a_{i_1i_2\cdots i_d}$ is determined by $d$ indices. We use the convention that indices start from 1, such that $1 \leq i_k \leq n_k\, (k=1,\ldots,d)$. The numbers $n_1,n_2,\ldots,n_d$ are called the dimensions of the tensor. A tensor is cubical if all its dimensions are equal. For practical purposes, only real tensors are considered. We use boldface capital calligraphic letters $\ten{A},\ten{B},\ldots$ to denote tensors, boldface capital letters $\mat{A},\mat{B},\ldots$ to denote matrices, boldface letters $\mat{a},\mat{b},\ldots$ to denote vectors, and Roman letters $a,b,\ldots$ to denote scalars. The elements of a set of $d$ tensors, in particular in the context of Tensor Networks, are denoted $\ten{A}^{(1)},\ten{A}^{(2)},\ldots,\ten{A}^{(d)}$. The transpose of a matrix $\mat{A}$ or vector $\mat{a}$ are denoted $\mat{A}^T$ and $\mat{a}^T$, respectively. The unit matrix of order $n$ is denoted $\mat{I}_n$. The tensor with all zero entries is denoted $\ten{O}$. We also adopt the Matlab notation $\textrm{diag}(\mat{a})$ for a diagonal matrix with entries $a_i$. Similarly, $\textrm{diag}(\mat{A})$ denotes the diagonal entries of a matrix $\mat{A}$.

Good introductions to tensors in scientific computing and signal processing are \cite{tensorreview,Cichocki2015}. The work in this article builds upon the tensor framework described in \cite{MVMALS}, in which a Tensor Network alternating linear scheme is derived for the identification of MIMO Volterra systems. Due to space limitation, we refer the reader to the discussion presented in \cite{MVMALS} on basic tensor operations. The same notation and concepts will be used in this article. Additional important tensor operations for this article that are not described in \cite{MVMALS} are given below.
\begin{definition}\cite[p.~462]{tensorreview}(Khatri-Rao product) Given matrices $\mat{A} \in \mathbb{R}^{n \times l}, \mat{B} \in \mathbb{R}^{m \times l}$, then their Khatri-Rao product $\mat{A} \odot \mat{B} \in \mathbb{R}^{nm \times l}$ is defined as the column-wise Kronecker product
\begin{align*}
\mat{A} \odot \mat{B} := \begin{pmatrix}\mat{A}(:,1) \otimes \mat{B}(:,1) &  \cdots & \mat{A}(:,l) \otimes \mat{B}(:,l)\end{pmatrix},
\end{align*}
where we used the Matlab-notation $\mat{A}(:,k)$ to denote the $k$th column of the matrix $\mat{A}$.
\end{definition}
The Khatri-Rao product of two matrices $\mat{A},\mat{B}$ hence corresponds with the matrix that contains the column-wise Kronecker product of $\mat{A}$ with $\mat{B}$. Similarly, an operation will be required where the Kronecker product is replaced with the outer product.
\begin{definition}Given matrices $\mat{A} \in \mathbb{R}^{n \times l}, \mat{B} \in \mathbb{R}^{m \times l}$, then their column-wise outer product $\mat{A}\, \square \, \mat{B} \in \mathbb{R}^{n\times m \times l}$ is defined as the column-wise outer product such that
\begin{align*}
\mat{A}\, \square \, \mat{B}(:,:,k) := \mat{A}(:,k) \circ \mat{B}(:,k),
\end{align*}
for $k=1,\ldots,l$ and where $\circ$ denotes the outer product\cite[p.~458]{tensorreview}.
\end{definition}
One can obtain $\mat{A}\, \square \, \mat{B}$ from reshaping the matrix $\mat{A}\, \odot \, \mat{B}$ into a 3-way tensor. Next we will provide the definition of the tensor Kronecker product, but before doing so, we first need to discuss multi-indices. A $d$-way tensor $\ten{A}$ is essentially a collection of numbers $a_{i_1\cdots i_d}$, and there are therefore many ways to arrange the entries. These different arrangements are equivalent with the grouping of the indices into separate groups. Consider the case where the first $k$ indices are grouped together into the multi-index $[i_1i_2\cdots i_k]$, keeping all remaining indices separate. This reduces the order of the tensor from $d$ to $d-k+1$. The multi-index $[i_1i_2\cdots i_k]$ is converted into a single linear index as
\begin{align*}
i_1 + (i_2-1) n_1 + \cdots + (i_k-1) n_1\,n_2\,\cdots \,n_{k-1}.
\end{align*}
A particular useful case occurs when $k=d$, for which the tensor $\ten{A}$ is reshaped into a vector, called the vectorization $\textrm{vec}(\ten{A})$, with entries indexed by $[i_1i_2\cdots i_d]$. With the notation of multi-indices explained, the definition of the tensor Kronecker product can be given.
\begin{definition}\cite{tkpsvd,Phan2012}(Tensor Kronecker product)
\label{def:tkron}
Let $\ten{B} \in \mathbb{R}^{n_1 \times n_2 \times \cdots \times n_d},\ten{C} \in \mathbb{R}^{m_1 \times m_2 \times \cdots \times m_d}$ be two $d$-way tensors with entries denoted by $b_{i_{d+1}\cdots i_{2d}}$ and $c_{i_1\cdots i_d}$, respectively. The tensor Kronecker product $\ten{A} = \ten{B} \otimes \ten{C} \in  \mathbb{R}^{n_1m_1 \times n_2m_2 \times \cdots \times n_km_k}$ is then defined from 
$$
a_{[i_1i_{d+1}][i_2i_{d+2}]\cdots [i_di_{2d}]} = b_{i_{d+1}\cdots i_{2d}}\,c_{i_1\cdots i_d},
$$
which needs to hold for all possible values of $i_1,\ldots,i_{2d}$.
\end{definition}

\subsection{Tensor Network theory}
The notion of Tensor Networks comes from physics, where they are used to describe the wave-function of entangled many-body quantum systems \cite{TNorus}. More recently, their application in big data processing has been proposed in \cite{Cichocki2014}. A Tensor Network (TN) is a set of tensors where some, or all, indices are contracted. These contractions constitute the network. An index contraction is the sum over all possible values of the repeated indices of a set of tensors. For example, the index contraction over the tensors $\ten{A},\ten{B},\ten{C}$
\begin{align}
d_{i_1i_3i_5i_6i_7} &= \sum_{i_2,i_4} a_{i_1i_2i_3} \; b_{i_2i_4} \; c_{i_4i_5i_6i_7}  
\label{eq:TNexample1}
\end{align}
results in a 5-way tensor $\ten{D}$. A particular useful representation of TNs is in terms of TN diagrams. In these diagrams, tensors are represented as nodes in a graph and the branches represent indices. Branches that are fully connected represent the contraction of the corresponding index. We will refer to the tensors inside a TN as the TN-cores and the dimensions of the contracted indices as the corresponding TN-ranks. The term TN-ranks is reserved in the literature for the minimal dimensions of the contracted indices such that the TN represents the underlying tensor exactly. Therefore, if one chooses dimensions for the contracted indices that are smaller than the TN-ranks, an approximation of the underlying tensor is obtained. The TN diagram of equation \eqref{eq:TNexample1} is shown in Figure \ref{fig:TN1}, where all indices are explicitly indicated. In this example, the contracted indices are $i_2,i_4$ and the corresponding TN-ranks are hence denoted $r_2,r_4$. The order of the resulting tensor in the TN can be quickly deduced from counting the number of ``free" indices in the diagram.
\begin{figure}[tb]
\begin{center}
\ifx\du\undefined
  \newlength{\du}
\fi
\setlength{\du}{5.5\unitlength}
\begin{tikzpicture}
\pgftransformxscale{1.000000}
\pgftransformyscale{-1.000000}
\definecolor{dialinecolor}{rgb}{0.000000, 0.000000, 0.000000}
\pgfsetstrokecolor{dialinecolor}
\definecolor{dialinecolor}{rgb}{1.000000, 1.000000, 1.000000}
\pgfsetfillcolor{dialinecolor}
\definecolor{dialinecolor}{rgb}{1.000000, 1.000000, 1.000000}
\pgfsetfillcolor{dialinecolor}
\pgfpathellipse{\pgfpoint{10.846636\du}{10.823318\du}}{\pgfpoint{3.453364\du}{0\du}}{\pgfpoint{0\du}{3.376682\du}}
\pgfusepath{fill}
\pgfsetlinewidth{0.100000\du}
\pgfsetdash{}{0pt}
\pgfsetdash{}{0pt}
\pgfsetmiterjoin
\definecolor{dialinecolor}{rgb}{0.000000, 0.000000, 0.000000}
\pgfsetstrokecolor{dialinecolor}
\pgfpathellipse{\pgfpoint{10.846636\du}{10.823318\du}}{\pgfpoint{3.453364\du}{0\du}}{\pgfpoint{0\du}{3.376682\du}}
\pgfusepath{stroke}
\definecolor{dialinecolor}{rgb}{0.000000, 0.000000, 0.000000}
\pgfsetstrokecolor{dialinecolor}
\node at (10.846636\du,11.509707\du){$\bm{\mathcal{A}}$};
\definecolor{dialinecolor}{rgb}{1.000000, 1.000000, 1.000000}
\pgfsetfillcolor{dialinecolor}
\pgfpathellipse{\pgfpoint{25.598364\du}{10.786682\du}}{\pgfpoint{3.453364\du}{0\du}}{\pgfpoint{0\du}{3.376682\du}}
\pgfusepath{fill}
\pgfsetlinewidth{0.100000\du}
\pgfsetdash{}{0pt}
\pgfsetdash{}{0pt}
\pgfsetmiterjoin
\definecolor{dialinecolor}{rgb}{0.000000, 0.000000, 0.000000}
\pgfsetstrokecolor{dialinecolor}
\pgfpathellipse{\pgfpoint{25.598364\du}{10.786682\du}}{\pgfpoint{3.453364\du}{0\du}}{\pgfpoint{0\du}{3.376682\du}}
\pgfusepath{stroke}
\definecolor{dialinecolor}{rgb}{0.000000, 0.000000, 0.000000}
\pgfsetstrokecolor{dialinecolor}
\node at (25.598364\du,11.473071\du){$\bm{\mathcal{B}}$};
\definecolor{dialinecolor}{rgb}{1.000000, 1.000000, 1.000000}
\pgfsetfillcolor{dialinecolor}
\pgfpathellipse{\pgfpoint{41.043364\du}{10.746682\du}}{\pgfpoint{3.453364\du}{0\du}}{\pgfpoint{0\du}{3.376682\du}}
\pgfusepath{fill}
\pgfsetlinewidth{0.100000\du}
\pgfsetdash{}{0pt}
\pgfsetdash{}{0pt}
\pgfsetmiterjoin
\definecolor{dialinecolor}{rgb}{0.000000, 0.000000, 0.000000}
\pgfsetstrokecolor{dialinecolor}
\pgfpathellipse{\pgfpoint{41.043364\du}{10.746682\du}}{\pgfpoint{3.453364\du}{0\du}}{\pgfpoint{0\du}{3.376682\du}}
\pgfusepath{stroke}
\definecolor{dialinecolor}{rgb}{0.000000, 0.000000, 0.000000}
\pgfsetstrokecolor{dialinecolor}
\node at (41.043364\du,11.433071\du){$\bm{\mathcal{C}}$};
\pgfsetlinewidth{0.100000\du}
\pgfsetdash{}{0pt}
\pgfsetdash{}{0pt}
\pgfsetbuttcap
{
\definecolor{dialinecolor}{rgb}{0.000000, 0.000000, 0.000000}
\pgfsetfillcolor{dialinecolor}
\definecolor{dialinecolor}{rgb}{0.000000, 0.000000, 0.000000}
\pgfsetstrokecolor{dialinecolor}
\draw (14.300000\du,10.823318\du)--(22.145000\du,10.786682\du);
}
\pgfsetlinewidth{0.100000\du}
\pgfsetdash{}{0pt}
\pgfsetdash{}{0pt}
\pgfsetbuttcap
{
\definecolor{dialinecolor}{rgb}{0.000000, 0.000000, 0.000000}
\pgfsetfillcolor{dialinecolor}
\definecolor{dialinecolor}{rgb}{0.000000, 0.000000, 0.000000}
\pgfsetstrokecolor{dialinecolor}
\draw (29.250000\du,10.750000\du)--(37.590000\du,10.746682\du);
}
\pgfsetlinewidth{0.100000\du}
\pgfsetdash{}{0pt}
\pgfsetdash{}{0pt}
\pgfsetbuttcap
{
\definecolor{dialinecolor}{rgb}{0.000000, 0.000000, 0.000000}
\pgfsetfillcolor{dialinecolor}
\definecolor{dialinecolor}{rgb}{0.000000, 0.000000, 0.000000}
\pgfsetstrokecolor{dialinecolor}
\draw (10.846636\du,14.200000\du)--(10.850000\du,17.300000\du);
}
\pgfsetlinewidth{0.100000\du}
\pgfsetdash{}{0pt}
\pgfsetdash{}{0pt}
\pgfsetbuttcap
{
\definecolor{dialinecolor}{rgb}{0.000000, 0.000000, 0.000000}
\pgfsetfillcolor{dialinecolor}
\definecolor{dialinecolor}{rgb}{0.000000, 0.000000, 0.000000}
\pgfsetstrokecolor{dialinecolor}
\draw (10.895054\du,4.360054\du)--(10.846636\du,7.446636\du);
}
\pgfsetlinewidth{0.100000\du}
\pgfsetdash{}{0pt}
\pgfsetdash{}{0pt}
\pgfsetbuttcap
{
\definecolor{dialinecolor}{rgb}{0.000000, 0.000000, 0.000000}
\pgfsetfillcolor{dialinecolor}
\definecolor{dialinecolor}{rgb}{0.000000, 0.000000, 0.000000}
\pgfsetstrokecolor{dialinecolor}
\draw (41.090054\du,4.270054\du)--(41.043364\du,7.370000\du);
}
\pgfsetlinewidth{0.100000\du}
\pgfsetdash{}{0pt}
\pgfsetdash{}{0pt}
\pgfsetbuttcap
{
\definecolor{dialinecolor}{rgb}{0.000000, 0.000000, 0.000000}
\pgfsetfillcolor{dialinecolor}
\definecolor{dialinecolor}{rgb}{0.000000, 0.000000, 0.000000}
\pgfsetstrokecolor{dialinecolor}
\draw (41.043364\du,14.123364\du)--(41.038418\du,17.230054\du);
}
\pgfsetlinewidth{0.100000\du}
\pgfsetdash{}{0pt}
\pgfsetdash{}{0pt}
\pgfsetbuttcap
{
\definecolor{dialinecolor}{rgb}{0.000000, 0.000000, 0.000000}
\pgfsetfillcolor{dialinecolor}
\definecolor{dialinecolor}{rgb}{0.000000, 0.000000, 0.000000}
\pgfsetstrokecolor{dialinecolor}
\draw (44.496728\du,10.746682\du)--(47.800000\du,10.700000\du);
}
\definecolor{dialinecolor}{rgb}{0.000000, 0.000000, 0.000000}
\pgfsetstrokecolor{dialinecolor}
\node[anchor=west] at (9.550000\du,19.500000\du){$i_1$};
\definecolor{dialinecolor}{rgb}{0.000000, 0.000000, 0.000000}
\pgfsetstrokecolor{dialinecolor}
\node[anchor=west] at (9.595000\du,3.207500\du){$i_3$};
\definecolor{dialinecolor}{rgb}{0.000000, 0.000000, 0.000000}
\pgfsetstrokecolor{dialinecolor}
\node[anchor=west] at (39.840000\du,19.767500\du){$i_5$};
\definecolor{dialinecolor}{rgb}{0.000000, 0.000000, 0.000000}
\pgfsetstrokecolor{dialinecolor}
\node[anchor=west] at (39.735000\du,3.127500\du){$i_6$};
\definecolor{dialinecolor}{rgb}{0.000000, 0.000000, 0.000000}
\pgfsetstrokecolor{dialinecolor}
\node[anchor=west] at (48.380000\du,11.287500\du){$i_7$};
\definecolor{dialinecolor}{rgb}{0.000000, 0.000000, 0.000000}
\pgfsetstrokecolor{dialinecolor}
\node[anchor=west] at (16.575000\du,9.897500\du){$i_2$};
\definecolor{dialinecolor}{rgb}{0.000000, 0.000000, 0.000000}
\pgfsetstrokecolor{dialinecolor}
\node[anchor=west] at (31.870000\du,9.757500\du){$i_4$};
\end{tikzpicture}
\caption{A TN of three tensors $\ten{A},\ten{B},\ten{C}$ resulting in a 5-way tensor.}
\label{fig:TN1}
\end{center}
\end{figure}
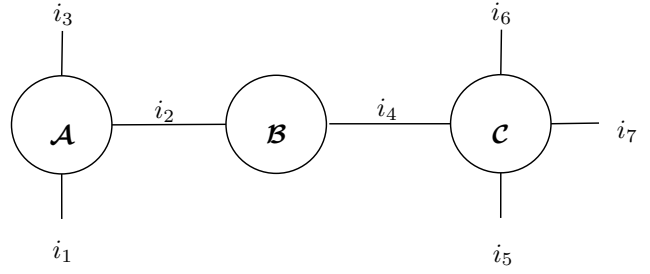
Two important TNs for the Kalman filter are the Tensor Train (TT)\cite{ivanTT} and the Tensor Train matrix (TTm)\cite{Oseledets2010}. A TT, represented by its TN diagram in Figure \ref{fig:TT}, is a TN of $d$ linearly connected cores $\ten{A}^{(1)},\ldots,\ten{A}^{(d)}$. The border cores are 2-way tensors (matrices) and the remaining cores are 3-way. A TTm extends the notion of a TT by increasing the order of each of the cores by 1.
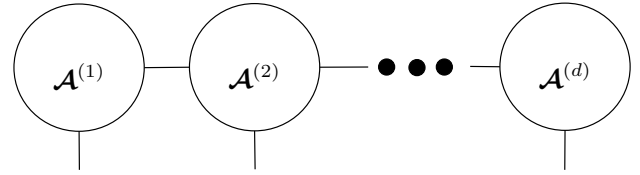
\begin{figure}[tb]
\begin{center}
\ifx\du\undefined
  \newlength{\du}
\fi
\setlength{\du}{4.5\unitlength}
\begin{tikzpicture}
\pgftransformxscale{1.000000}
\pgftransformyscale{-1.000000}
\definecolor{dialinecolor}{rgb}{0.000000, 0.000000, 0.000000}
\pgfsetstrokecolor{dialinecolor}
\definecolor{dialinecolor}{rgb}{1.000000, 1.000000, 1.000000}
\pgfsetfillcolor{dialinecolor}
\definecolor{dialinecolor}{rgb}{1.000000, 1.000000, 1.000000}
\pgfsetfillcolor{dialinecolor}
\pgfpathellipse{\pgfpoint{8.848364\du}{12.803318\du}}{\pgfpoint{5.475129\du}{0\du}}{\pgfpoint{0\du}{5.353554\du}}
\pgfusepath{fill}
\pgfsetlinewidth{0.100000\du}
\pgfsetdash{}{0pt}
\pgfsetdash{}{0pt}
\pgfsetmiterjoin
\definecolor{dialinecolor}{rgb}{0.000000, 0.000000, 0.000000}
\pgfsetstrokecolor{dialinecolor}
\pgfpathellipse{\pgfpoint{8.848364\du}{12.803318\du}}{\pgfpoint{5.475129\du}{0\du}}{\pgfpoint{0\du}{5.353554\du}}
\pgfusepath{stroke}
\definecolor{dialinecolor}{rgb}{0.000000, 0.000000, 0.000000}
\pgfsetstrokecolor{dialinecolor}
\node at (8.848364\du,13.489707\du){$\bm{\mathcal{A}}^{(1)}$};
\definecolor{dialinecolor}{rgb}{1.000000, 1.000000, 1.000000}
\pgfsetfillcolor{dialinecolor}
\pgfpathellipse{\pgfpoint{23.600092\du}{12.766682\du}}{\pgfpoint{5.475130\du}{0\du}}{\pgfpoint{0\du}{5.353555\du}}
\pgfusepath{fill}
\pgfsetlinewidth{0.100000\du}
\pgfsetdash{}{0pt}
\pgfsetdash{}{0pt}
\pgfsetmiterjoin
\definecolor{dialinecolor}{rgb}{0.000000, 0.000000, 0.000000}
\pgfsetstrokecolor{dialinecolor}
\pgfpathellipse{\pgfpoint{23.600092\du}{12.766682\du}}{\pgfpoint{5.475130\du}{0\du}}{\pgfpoint{0\du}{5.353555\du}}
\pgfusepath{stroke}
\definecolor{dialinecolor}{rgb}{0.000000, 0.000000, 0.000000}
\pgfsetstrokecolor{dialinecolor}
\node at (23.600092\du,13.453071\du){$\bm{\mathcal{A}}^{(2)}$};
\definecolor{dialinecolor}{rgb}{1.000000, 1.000000, 1.000000}
\pgfsetfillcolor{dialinecolor}
\pgfpathellipse{\pgfpoint{49.695092\du}{12.726682\du}}{\pgfpoint{5.472776\du}{0\du}}{\pgfpoint{0\du}{5.351253\du}}
\pgfusepath{fill}
\pgfsetlinewidth{0.100000\du}
\pgfsetdash{}{0pt}
\pgfsetdash{}{0pt}
\pgfsetmiterjoin
\definecolor{dialinecolor}{rgb}{0.000000, 0.000000, 0.000000}
\pgfsetstrokecolor{dialinecolor}
\pgfpathellipse{\pgfpoint{49.695092\du}{12.726682\du}}{\pgfpoint{5.472776\du}{0\du}}{\pgfpoint{0\du}{5.351253\du}}
\pgfusepath{stroke}
\definecolor{dialinecolor}{rgb}{0.000000, 0.000000, 0.000000}
\pgfsetstrokecolor{dialinecolor}
\node at (49.695092\du,13.413071\du){$\bm{\mathcal{A}}^{(d)}$};
\pgfsetlinewidth{0.100000\du}
\pgfsetdash{}{0pt}
\pgfsetdash{}{0pt}
\pgfsetbuttcap
{
\definecolor{dialinecolor}{rgb}{0.000000, 0.000000, 0.000000}
\pgfsetfillcolor{dialinecolor}
\definecolor{dialinecolor}{rgb}{0.000000, 0.000000, 0.000000}
\pgfsetstrokecolor{dialinecolor}
\draw (14.323493\du,12.803318\du)--(18.124962\du,12.766682\du);
}
\pgfsetlinewidth{0.100000\du}
\pgfsetdash{}{0pt}
\pgfsetdash{}{0pt}
\pgfsetbuttcap
{
\definecolor{dialinecolor}{rgb}{0.000000, 0.000000, 0.000000}
\pgfsetfillcolor{dialinecolor}
\definecolor{dialinecolor}{rgb}{0.000000, 0.000000, 0.000000}
\pgfsetstrokecolor{dialinecolor}
\draw (29.075222\du,12.766682\du)--(33.150000\du,12.750000\du);
}
\pgfsetlinewidth{0.100000\du}
\pgfsetdash{}{0pt}
\pgfsetdash{}{0pt}
\pgfsetbuttcap
{
\definecolor{dialinecolor}{rgb}{0.000000, 0.000000, 0.000000}
\pgfsetfillcolor{dialinecolor}
\definecolor{dialinecolor}{rgb}{0.000000, 0.000000, 0.000000}
\pgfsetstrokecolor{dialinecolor}
\draw (8.848364\du,18.156872\du)--(8.873235\du,21.400000\du);
}
\pgfsetlinewidth{0.100000\du}
\pgfsetdash{}{0pt}
\pgfsetdash{}{0pt}
\pgfsetbuttcap
{
\definecolor{dialinecolor}{rgb}{0.000000, 0.000000, 0.000000}
\pgfsetfillcolor{dialinecolor}
\definecolor{dialinecolor}{rgb}{0.000000, 0.000000, 0.000000}
\pgfsetstrokecolor{dialinecolor}
\draw (41.723235\du,12.700000\du)--(44.222316\du,12.726682\du);
}
\definecolor{dialinecolor}{rgb}{0.000000, 0.000000, 0.000000}
\pgfsetfillcolor{dialinecolor}
\pgfpathellipse{\pgfpoint{34.700000\du}{12.750000\du}}{\pgfpoint{0.675000\du}{0\du}}{\pgfpoint{0\du}{0.675000\du}}
\pgfusepath{fill}
\pgfsetlinewidth{0.100000\du}
\pgfsetdash{}{0pt}
\pgfsetdash{}{0pt}
\definecolor{dialinecolor}{rgb}{0.000000, 0.000000, 0.000000}
\pgfsetstrokecolor{dialinecolor}
\pgfpathellipse{\pgfpoint{34.700000\du}{12.750000\du}}{\pgfpoint{0.675000\du}{0\du}}{\pgfpoint{0\du}{0.675000\du}}
\pgfusepath{stroke}
\definecolor{dialinecolor}{rgb}{0.000000, 0.000000, 0.000000}
\pgfsetfillcolor{dialinecolor}
\pgfpathellipse{\pgfpoint{37.250000\du}{12.800000\du}}{\pgfpoint{0.675000\du}{0\du}}{\pgfpoint{0\du}{0.675000\du}}
\pgfusepath{fill}
\pgfsetlinewidth{0.100000\du}
\pgfsetdash{}{0pt}
\pgfsetdash{}{0pt}
\definecolor{dialinecolor}{rgb}{0.000000, 0.000000, 0.000000}
\pgfsetstrokecolor{dialinecolor}
\pgfpathellipse{\pgfpoint{37.250000\du}{12.800000\du}}{\pgfpoint{0.675000\du}{0\du}}{\pgfpoint{0\du}{0.675000\du}}
\pgfusepath{stroke}
\definecolor{dialinecolor}{rgb}{0.000000, 0.000000, 0.000000}
\pgfsetfillcolor{dialinecolor}
\pgfpathellipse{\pgfpoint{39.550000\du}{12.750000\du}}{\pgfpoint{0.675000\du}{0\du}}{\pgfpoint{0\du}{0.675000\du}}
\pgfusepath{fill}
\pgfsetlinewidth{0.100000\du}
\pgfsetdash{}{0pt}
\pgfsetdash{}{0pt}
\definecolor{dialinecolor}{rgb}{0.000000, 0.000000, 0.000000}
\pgfsetstrokecolor{dialinecolor}
\pgfpathellipse{\pgfpoint{39.550000\du}{12.750000\du}}{\pgfpoint{0.675000\du}{0\du}}{\pgfpoint{0\du}{0.675000\du}}
\pgfusepath{stroke}
\pgfsetlinewidth{0.100000\du}
\pgfsetdash{}{0pt}
\pgfsetdash{}{0pt}
\pgfsetbuttcap
{
\definecolor{dialinecolor}{rgb}{0.000000, 0.000000, 0.000000}
\pgfsetfillcolor{dialinecolor}
\definecolor{dialinecolor}{rgb}{0.000000, 0.000000, 0.000000}
\pgfsetstrokecolor{dialinecolor}
\draw (23.600092\du,18.120237\du)--(23.643487\du,21.353510\du);
}
\pgfsetlinewidth{0.100000\du}
\pgfsetdash{}{0pt}
\pgfsetdash{}{0pt}
\pgfsetbuttcap
{
\definecolor{dialinecolor}{rgb}{0.000000, 0.000000, 0.000000}
\pgfsetfillcolor{dialinecolor}
\definecolor{dialinecolor}{rgb}{0.000000, 0.000000, 0.000000}
\pgfsetstrokecolor{dialinecolor}
\draw (49.695092\du,18.077935\du)--(49.673235\du,21.400000\du);
}
\end{tikzpicture}
\caption{A TT of $d$ linearly connected cores $\ten{A}^{(1)},\ldots,\ten{A}^{(d)}$ resulting in a $d$-way tensor.}
\label{fig:TT}
\end{center}
\end{figure}
The notion of a TN allows for a compact representation of a given tensor, thus avoiding the potential curse of dimensionality. The following example illustrates how a TT can be used to store a vector of exponential length.
\begin{example}
Suppose we have a vector $\mat{a}$ of length $10^{10}$. Each entry of $\mat{a}$ can be indexed by a multi-index $[i_1,\ldots,i_{10}]$. We now reshape $\mat{a}$ into a $10 \times 10 \times \cdots \times 10$ tensor $\ten{A}$ such that $ a_{i_1\cdots i_{10}} := a_{[i_1\cdots i_{10}]}$. The TN representation of $\ten{A}$ then consists of 10 TN-cores $\ten{A}^{(1)},\ldots,\ten{A}^{(10)}$ with TN-ranks $r_1,\ldots,r_9$. If we denote the maximal TN-rank by $r$ then the total storage requirement for the vector $\mat{a}$ is reduced from $10^{10}$ to $O(100\,r^2)$, which can be a significant reduction when $r$ is small. This conversion into the TT format assigns each index $i_k$ of the multi-index to the core $\ten{A}^{(k)}$.
\end{example}
Similarly, the TTm format represents an $n^d \times n^d$ matrix, where each TTm-core $\ten{A}^{(k)}$ now is a 4-way tensor, with 2 free indices $i_k,j_k$, one denoting a row index and one denoting a column index, respectively. Using the TTm format reduces the storage requirement from $n^{2d}$ to $O(dn^2r^2)$. In the context of MIMO Volterra systems, a modification of the TT and TTm concepts will be made where the order of the border tensor $\ten{A}^{(1)}$ is increased by one. This extension is further explained in Section \ref{sec:VolterraTensor}. In the following two sections both the derivation and implementation of the TN Kalman filter is given.

\section{Tensor Network Kalman filter}
\label{sec:kalman}
Consider the following linear discrete-time state space model
\begin{align}
\nonumber \mat{X}(t+1) &= \mat{A}(t)\, \mat{X}(t) + \mat{W}(t),\\ 
\mat{y}(t) &= \mat{c}(t) \, \mat{X}(t) + \mat{r}(t),
\label{eqn:linearstatespace}
\end{align}
where $\mat{X}(t) \in \mathbb{R}^{n^d \times l}$ is the matrix containing $l$ exponentially long state vectors, $\mat{y}(t) \in \mathbb{R}^{1 \times l}$ is a vector of $l$ measurements, $\mat{A}(t) \in \mathbb{R}^{n^d \times n^d }$ is the state transition matrix, $\mat{c}(t) \in \mathbb{R}^{1 \times n^d}$ converts the state vectors into measurements and $\mat{W}(t) \in \mathbb{R}^{n^d \times l}, \mat{r}(t) \in \mathbb{R}^{1 \times l}$ denote process and measurement noise, respectively. It is possible to generalize $\mat{c}(t)$ to a matrix, thus allowing for matrices of measurements but we leave this for future work. For compactness the time-dependence in the notation is removed and the following assumptions are made:
\begin{itemize}
\item Each column $\mat{x}_k \,(k=1,\ldots,l)$ of the matrix $\mat{X}(0)$ follows a multivariate Gaussian distribution
\begin{align*}
\textrm{N}(\mat{m}_k,\mat{P}_k) = \frac{1}{Z}\, \textrm{exp}\left( -\frac{1}{2} (\mat{x}_k-\mat{m}_k)^T \, \mat{P}_k \, (\mat{x}_k-\mat{m}_k) \right ),
\end{align*}
with normalization constant $Z:=((2\,\pi)^{n^d/2}\,|\mat{P}_k|^{1/2})^{-1}$, where $|\mat{P}_k|$ denotes the determinant of $\mat{P}_k$. The vectors $\mat{m}_k$ are collected in the matrix $\mat{M} \in \mathbb{R}^{n^d \times l}$ and similarly all covariance matrices are collected into a 3-way tensor $\ten{P} \in \mathbb{R}^{n^d \times n^d \times l}$,
\item each column of the process noise matrix $\mat{W}$ is a multivariate Gaussian white noise process. This implies zero means and diagonal covariance matrices, which are collected into a 3-way tensor $\ten{Q} \in \mathbb{R}^{n^d \times n^d \times l}$,
\item the measurement noise $\mat{r}$ is a multivariate Gaussian white noise process with diagonal covariance matrix $\mat{R} \in \mathbb{R}^{l\times l}$. The row vector containing the diagonal entries of $\mat{R}$ is denoted $\textrm{diag}(\mat{R})$. 
\end{itemize}
In Bayesian filtering, one is interested in computing the distribution of the current state given the current and all previous measurements of the output $p(\mat{x}_k(t) | \mat{y}(1),\ldots,\mat{y}(t) )$. The linearity of the state space model together with a Markov-Gaussian assumption of the distributions results in the Kalman filter equations as the solution to the Bayesian filtering problem for each of the columns of $\mat{X}$~\cite[p.~56]{bayesianfiltering}. However, it is not necessary to run a Kalman filter for each of the $l$ columns of $\mat{X}$ separately. The whole Bayesian filtering problem can be solved with one Kalman filter, where both the prediction and update steps are rewritten as tensor equations. We reintroduce the time-dependence in the notation and denote the matrix of predicted means $\mat{M}(t)$ and tensor of predicted covariance matrices $\ten{P}(t)$ by $\mat{M}^+$ and $\ten{P}^+$, respectively. The tensor-times-matrix $k$-mode product is denoted by $\times_k$~\cite[p.~460]{tensorreview}. The Kalman filter prediction step is then rewritten as
\begin{align}
 \mat{M}^+ &= \mat{M}(t-1) \times_1 \mat{A}(t-1), \nonumber \\
\ten{P}^+&= \ten{P}(t-1) \times_1 \mat{A}(t-1) \times_2 \mat{A}(t-1) + \ten{Q}(t-1).
\label{eqns:predict}
\end{align}
Similarly, the Kalman filter update step is rewritten as
\begin{align}
\nonumber \mat{v} &= \mat{y}(t) - \mat{M}^+ \times_1 \mat{c}(t),\\
\nonumber \mat{s} &= \ten{P}^+\times_1 \mat{c}(t) \times_2 \mat{c}(t) + \textrm{diag}(\mat{R}(t)),\\
\nonumber\mat{K} &= \ten{P}^+\times_2 \mat{c}(t) \times_3 \textrm{diag}(\mat{s})^{-1},\\
\nonumber \mat{M}(t) &= \mat{M}^+ + \mat{K}\times_2 \textrm{diag}(\mat{v}),\\
 \ten{P}(t) &= \ten{P}^+ - (\mat{K} \, \square \, \mat{K})\times_3 \textrm{diag}(\mat{s}).
\label{eqns:update}
\end{align}
The prediction and update equations \eqref{eqns:predict} and \eqref{eqns:update} reduce to the standard matrix equations of a Kalman filter for a scalar output $y(t)$ when $l=1$. Both the mean $\mat{M}$ and the Kalman gain $\mat{K}$ reduce to column vectors for this case. Note that $\textrm{diag}(\mat{s})$ and $\textrm{diag}(\mat{v})$ both denote diagonal matrices containing the entries of $\mat{s},\mat{v}$, respectively. The notation $\textrm{diag}(\mat{R}(t))$ on the other hand denotes a row vector, containing the diagonal entries of the matrix $\mat{R}(t)$. The dimensions of the matrices and tensors in both \eqref{eqns:predict} and \eqref{eqns:update} suffer from the curse of dimensionality. The Kalman filter steps can therefore only be computed for small values of $n$ and $d$. In the next section we show how each iteration of this tensor Kalman filter can be efficiently computed using TNs, enabling much larger values of $n$ and $d$. In practice, one would use a square-root filter for improved numerical conditioning. Similarly to \eqref{eqns:predict} and \eqref{eqns:update}, one can rewrite the square root predict and update equations as tensor equations and compute them in the TN format. This will be reported in our future work. 

\section{Implementation}
\label{sec:implementation}
In this section we explain how each tensor equation of the tensor Kalman filter \eqref{eqns:predict} and \eqref{eqns:update} is computed in terms of TNs. The key idea is here that all computations are done in terms of TNs and all results stay in the TN format. This means that at no point in time the matrices and tensors in \eqref{eqns:predict} and \eqref{eqns:update} are explicitly constructed. Starting the computations requires the matrices $\mat{M}(0),\mat{A}(0)$ and the tensors $\ten{P}(0),\ten{Q}(0)$ in the TN format. It is therefore crucial to be able to initialize these TNs efficiently. The following additional assumptions are therefore made:
\begin{itemize}
\item $l \ll n^d$,
\item the matrix of initial mean vectors $\mat{M}(0)$ is the zero matrix,
\item each of the $l$ covariance matrices inside $\ten{P}(0)$ is a diagonal matrix with a constant value $\sigma_k^2 \,(k=1,\ldots,l)$ on the diagonal,
\item the matrix $\mat{A}(t)$ is given in the TTm format,
\item the row vector $\mat{c}(t)$ is given in the TT format.
\end{itemize}
The TN representation of the $n^d \times l$ matrix $\mat{M}(t)$ can be thought of as a TT where the first core $\ten{M}^{(1)}(t)$ is a $l \times n \times r_1$ tensor. The covariance tensor $\ten{P}(t)$ can be represented by a TN similar to the TTm format where the first core is an $l\times n \times n \times R_1$ tensor. The following lemmas explain how $\mat{M}(0)$ and $\ten{P}(0)$ can be directly initialized in their corresponding TN formats when the previously mentioned assumptions hold.
\begin{lemma}
All $d$ TN-cores $\ten{M}^{(1)},\ldots,\ten{M}^{(d)}$ of $\mat{M}(0)$ have unit TN-ranks with
\begin{align*}
\ten{M}^{(1)} &= \ten{O} \in \mathbb{R}^{l \times n \times 1},\\
\ten{M}^{(k)} &= \ten{O} \in \mathbb{R}^{1 \times n \times 1} \quad (k=2,\ldots,d).
\end{align*}
\label{lemma:M0}
\end{lemma}
The total storage cost to store $\mat{M}(0)$ is hence reduced from $n^d$ to $(l+d-1)n$. The next lemma explains how the TN of $\ten{P}(0)$ can be initialized.
\begin{lemma}
The TN of $\ten{P}(0)$ consists of the following TN-cores with all unit TN-ranks
\begin{align*}
\ten{P}^{(1)}(i,:,:,1) &= \sigma_{i}^2 \, \mat{I}_{n} \; (i=1,\ldots,l),\\
\ten{P}^{(k)}(1,:,:,1)  &= \mat{I}_{n} \quad (k=2,\ldots,d).
\end{align*}
\label{lemma:P0}
\end{lemma}
The total storage cost for $\ten{P}(0)$ is in this way reduced from $n^{2d}$ to $(l+d-1)n^2$. Lemma \ref{lemma:P0} can also be used to construct the tensor $\ten{Q}(t)$ directly into the TN format. The proofs of Lemmas \ref{lemma:M0} and \ref{lemma:P0} rely on the fact that a TN with all unit TN-ranks is equivalent with the outer product of the TN-cores.  Table \ref{tab:storage} compares the storage cost between the standard and the TN approach where the maximal TN-ranks of $\mat{M}(t),\ten{P}(t),\mat{A}(t),\mat{c}$ are denoted $r_M,r_P,r_A,r_c$, respectively. From Table \ref{tab:storage} one can see that using TNs transforms the storage cost from exponential into linear in the degree $d$. When the above assumptions do not hold, then alternative methods to initialize the TNs are required. These are the TT-SVD~\cite[p.~2301]{ivanTT} and TT-cross approximation~\cite[p.~80]{ttcross} algorithms. We now discuss the implementation of each step of the TN Kalman filter.
\begin{table}[tb]
\begin{center}
\caption{Storage cost for the standard and TN approach.}
\label{tab:storage}	
\begin{tabular}{@{}lcc@{}}
		 &  \multicolumn{2}{c}{Storage cost} \\
\midrule
 Storage		& Standard		& TN \\ \midrule
$\mat{M}(t)$ 		& $ln^d$ 	& $O((d-1)nr_M^2+lnr_M)$  \\
$\ten{P}(t)$ 		& $n^{2d}$ 	& $O((d-1)n^2r_P^2+ln^2r_P)$\\
$\ten{Q}(t)$ 		& $n^{2d}$ 	& $(l+d-1)n^2$\\
$\mat{A}(t)$ 		& $n^{2d}$ 	& $O((d-1)n^2r_A^2+lnr_A)$\\
$\mat{c}$ 			& $n^{d}$ 	& $O(dnr_c^2)$\\
\end{tabular}
\end{center}
\end{table}

\subsection{TN implementation of $\mat{M}^+ = \mat{M}(t-1) \times_1 \mat{A}(t-1)$}
\begin{figure}[tb]
\begin{center}
\input{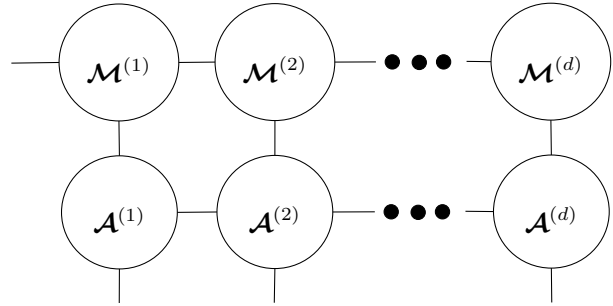}
\caption{The operation $\mat{M}(t-1) \times_1 \mat{A}(t-1)$ as a TN.}
\label{fig:TNMA}
\end{center}
\end{figure}
\begin{figure}[tb]
\begin{center}
\ifx\du\undefined
  \newlength{\du}
\fi
\setlength{\du}{4\unitlength}
\begin{tikzpicture}
\pgftransformxscale{1.000000}
\pgftransformyscale{-1.000000}
\definecolor{dialinecolor}{rgb}{0.000000, 0.000000, 0.000000}
\pgfsetstrokecolor{dialinecolor}
\definecolor{dialinecolor}{rgb}{1.000000, 1.000000, 1.000000}
\pgfsetfillcolor{dialinecolor}
\definecolor{dialinecolor}{rgb}{1.000000, 1.000000, 1.000000}
\pgfsetfillcolor{dialinecolor}
\pgfpathellipse{\pgfpoint{8.689118\du}{8.630829\du}}{\pgfpoint{5.659933\du}{0\du}}{\pgfpoint{0\du}{5.534254\du}}
\pgfusepath{fill}
\pgfsetlinewidth{0.100000\du}
\pgfsetdash{}{0pt}
\pgfsetdash{}{0pt}
\pgfsetmiterjoin
\definecolor{dialinecolor}{rgb}{0.000000, 0.000000, 0.000000}
\pgfsetstrokecolor{dialinecolor}
\pgfpathellipse{\pgfpoint{8.689118\du}{8.630829\du}}{\pgfpoint{5.659933\du}{0\du}}{\pgfpoint{0\du}{5.534254\du}}
\pgfusepath{stroke}
\definecolor{dialinecolor}{rgb}{0.000000, 0.000000, 0.000000}
\pgfsetstrokecolor{dialinecolor}
\node at (8.689118\du,9.317217\du){$\ten{M}^{+(1)}$};
\definecolor{dialinecolor}{rgb}{1.000000, 1.000000, 1.000000}
\pgfsetfillcolor{dialinecolor}
\pgfpathellipse{\pgfpoint{23.440918\du}{8.630829\du}}{\pgfpoint{5.659933\du}{0\du}}{\pgfpoint{0\du}{5.534254\du}}
\pgfusepath{fill}
\pgfsetlinewidth{0.100000\du}
\pgfsetdash{}{0pt}
\pgfsetdash{}{0pt}
\pgfsetmiterjoin
\definecolor{dialinecolor}{rgb}{0.000000, 0.000000, 0.000000}
\pgfsetstrokecolor{dialinecolor}
\pgfpathellipse{\pgfpoint{23.440918\du}{8.630829\du}}{\pgfpoint{5.659933\du}{0\du}}{\pgfpoint{0\du}{5.534254\du}}
\pgfusepath{stroke}
\definecolor{dialinecolor}{rgb}{0.000000, 0.000000, 0.000000}
\pgfsetstrokecolor{dialinecolor}
\node at (23.440918\du,9.317217\du){$\ten{M}^{+(2)}$};
\definecolor{dialinecolor}{rgb}{1.000000, 1.000000, 1.000000}
\pgfsetfillcolor{dialinecolor}
\pgfpathellipse{\pgfpoint{53.825975\du}{8.628515\du}}{\pgfpoint{5.657567\du}{0\du}}{\pgfpoint{0\du}{5.531941\du}}
\pgfusepath{fill}
\pgfsetlinewidth{0.100000\du}
\pgfsetdash{}{0pt}
\pgfsetdash{}{0pt}
\pgfsetmiterjoin
\definecolor{dialinecolor}{rgb}{0.000000, 0.000000, 0.000000}
\pgfsetstrokecolor{dialinecolor}
\pgfpathellipse{\pgfpoint{53.825975\du}{8.628515\du}}{\pgfpoint{5.657567\du}{0\du}}{\pgfpoint{0\du}{5.531941\du}}
\pgfusepath{stroke}
\definecolor{dialinecolor}{rgb}{0.000000, 0.000000, 0.000000}
\pgfsetstrokecolor{dialinecolor}
\node at (53.825975\du,9.314904\du){$\ten{M}^{+(d)}$};
\pgfsetlinewidth{0.100000\du}
\pgfsetdash{}{0pt}
\pgfsetdash{}{0pt}
\pgfsetbuttcap
{
\definecolor{dialinecolor}{rgb}{0.000000, 0.000000, 0.000000}
\pgfsetfillcolor{dialinecolor}
\definecolor{dialinecolor}{rgb}{0.000000, 0.000000, 0.000000}
\pgfsetstrokecolor{dialinecolor}
\draw (13.918214\du,10.748696\du)--(18.211821\du,10.748696\du);
}
\pgfsetlinewidth{0.100000\du}
\pgfsetdash{}{0pt}
\pgfsetdash{}{0pt}
\pgfsetbuttcap
{
\definecolor{dialinecolor}{rgb}{0.000000, 0.000000, 0.000000}
\pgfsetfillcolor{dialinecolor}
\definecolor{dialinecolor}{rgb}{0.000000, 0.000000, 0.000000}
\pgfsetstrokecolor{dialinecolor}
\draw (8.689118\du,14.165083\du)--(8.734084\du,17.538902\du);
}
\pgfsetlinewidth{0.100000\du}
\pgfsetdash{}{0pt}
\pgfsetdash{}{0pt}
\pgfsetbuttcap
{
\definecolor{dialinecolor}{rgb}{0.000000, 0.000000, 0.000000}
\pgfsetfillcolor{dialinecolor}
\definecolor{dialinecolor}{rgb}{0.000000, 0.000000, 0.000000}
\pgfsetstrokecolor{dialinecolor}
\draw (23.440918\du,14.165083\du)--(23.435884\du,17.502302\du);
}
\pgfsetlinewidth{0.100000\du}
\pgfsetdash{}{0pt}
\pgfsetdash{}{0pt}
\pgfsetbuttcap
{
\definecolor{dialinecolor}{rgb}{0.000000, 0.000000, 0.000000}
\pgfsetfillcolor{dialinecolor}
\definecolor{dialinecolor}{rgb}{0.000000, 0.000000, 0.000000}
\pgfsetstrokecolor{dialinecolor}
\draw (53.825975\du,14.160455\du)--(53.848220\du,17.017563\du);
}
\pgfsetlinewidth{0.100000\du}
\pgfsetdash{}{0pt}
\pgfsetdash{}{0pt}
\pgfsetbuttcap
{
\definecolor{dialinecolor}{rgb}{0.000000, 0.000000, 0.000000}
\pgfsetfillcolor{dialinecolor}
\definecolor{dialinecolor}{rgb}{0.000000, 0.000000, 0.000000}
\pgfsetstrokecolor{dialinecolor}
\draw (-0.707210\du,8.581130\du)--(3.029184\du,8.630829\du);
}
\pgfsetlinewidth{0.100000\du}
\pgfsetdash{}{0pt}
\pgfsetdash{}{0pt}
\pgfsetbuttcap
{
\definecolor{dialinecolor}{rgb}{0.000000, 0.000000, 0.000000}
\pgfsetfillcolor{dialinecolor}
\definecolor{dialinecolor}{rgb}{0.000000, 0.000000, 0.000000}
\pgfsetstrokecolor{dialinecolor}
\draw (13.918214\du,6.512961\du)--(18.211821\du,6.512961\du);
}
\pgfsetlinewidth{0.100000\du}
\pgfsetdash{}{0pt}
\pgfsetdash{}{0pt}
\pgfsetbuttcap
{
\definecolor{dialinecolor}{rgb}{0.000000, 0.000000, 0.000000}
\pgfsetfillcolor{dialinecolor}
\definecolor{dialinecolor}{rgb}{0.000000, 0.000000, 0.000000}
\pgfsetstrokecolor{dialinecolor}
\draw (28.670014\du,6.512961\du)--(33.524595\du,6.512961\du);
}
\pgfsetlinewidth{0.100000\du}
\pgfsetdash{}{0pt}
\pgfsetdash{}{0pt}
\pgfsetbuttcap
{
\definecolor{dialinecolor}{rgb}{0.000000, 0.000000, 0.000000}
\pgfsetfillcolor{dialinecolor}
\definecolor{dialinecolor}{rgb}{0.000000, 0.000000, 0.000000}
\pgfsetstrokecolor{dialinecolor}
\draw (43.982787\du,6.512961\du)--(48.599065\du,6.511533\du);
}
\pgfsetlinewidth{0.100000\du}
\pgfsetdash{}{0pt}
\pgfsetdash{}{0pt}
\pgfsetbuttcap
{
\definecolor{dialinecolor}{rgb}{0.000000, 0.000000, 0.000000}
\pgfsetfillcolor{dialinecolor}
\definecolor{dialinecolor}{rgb}{0.000000, 0.000000, 0.000000}
\pgfsetstrokecolor{dialinecolor}
\draw (28.670014\du,10.748696\du)--(33.524595\du,10.748696\du);
}
\pgfsetlinewidth{0.100000\du}
\pgfsetdash{}{0pt}
\pgfsetdash{}{0pt}
\pgfsetbuttcap
{
\definecolor{dialinecolor}{rgb}{0.000000, 0.000000, 0.000000}
\pgfsetfillcolor{dialinecolor}
\definecolor{dialinecolor}{rgb}{0.000000, 0.000000, 0.000000}
\pgfsetstrokecolor{dialinecolor}
\draw (43.982787\du,10.748696\du)--(48.599065\du,10.745497\du);
}
\pgfsetlinewidth{0.100000\du}
\pgfsetdash{}{0pt}
\pgfsetdash{}{0pt}
\pgfsetmiterjoin
\definecolor{dialinecolor}{rgb}{1.000000, 1.000000, 1.000000}
\pgfsetstrokecolor{dialinecolor}
\pgfpathellipse{\pgfpoint{38.753691\du}{8.630829\du}}{\pgfpoint{5.659933\du}{0\du}}{\pgfpoint{0\du}{5.534254\du}}
\pgfusepath{stroke}
\definecolor{dialinecolor}{rgb}{0.000000, 0.000000, 0.000000}
\pgfsetstrokecolor{dialinecolor}
\node at (38.753691\du,9.317217\du){};
\definecolor{dialinecolor}{rgb}{0.000000, 0.000000, 0.000000}
\pgfsetfillcolor{dialinecolor}
\pgfpathellipse{\pgfpoint{36.584249\du}{8.494143\du}}{\pgfpoint{0.675000\du}{0\du}}{\pgfpoint{0\du}{0.675000\du}}
\pgfusepath{fill}
\pgfsetlinewidth{0.100000\du}
\pgfsetdash{}{0pt}
\pgfsetdash{}{0pt}
\definecolor{dialinecolor}{rgb}{0.000000, 0.000000, 0.000000}
\pgfsetstrokecolor{dialinecolor}
\pgfpathellipse{\pgfpoint{36.584249\du}{8.494143\du}}{\pgfpoint{0.675000\du}{0\du}}{\pgfpoint{0\du}{0.675000\du}}
\pgfusepath{stroke}
\definecolor{dialinecolor}{rgb}{0.000000, 0.000000, 0.000000}
\pgfsetfillcolor{dialinecolor}
\pgfpathellipse{\pgfpoint{39.134249\du}{8.494143\du}}{\pgfpoint{0.675000\du}{0\du}}{\pgfpoint{0\du}{0.675000\du}}
\pgfusepath{fill}
\pgfsetlinewidth{0.100000\du}
\pgfsetdash{}{0pt}
\pgfsetdash{}{0pt}
\definecolor{dialinecolor}{rgb}{0.000000, 0.000000, 0.000000}
\pgfsetstrokecolor{dialinecolor}
\pgfpathellipse{\pgfpoint{39.134249\du}{8.494143\du}}{\pgfpoint{0.675000\du}{0\du}}{\pgfpoint{0\du}{0.675000\du}}
\pgfusepath{stroke}
\definecolor{dialinecolor}{rgb}{0.000000, 0.000000, 0.000000}
\pgfsetfillcolor{dialinecolor}
\pgfpathellipse{\pgfpoint{41.434249\du}{8.494143\du}}{\pgfpoint{0.675000\du}{0\du}}{\pgfpoint{0\du}{0.675000\du}}
\pgfusepath{fill}
\pgfsetlinewidth{0.100000\du}
\pgfsetdash{}{0pt}
\pgfsetdash{}{0pt}
\definecolor{dialinecolor}{rgb}{0.000000, 0.000000, 0.000000}
\pgfsetstrokecolor{dialinecolor}
\pgfpathellipse{\pgfpoint{41.434249\du}{8.494143\du}}{\pgfpoint{0.675000\du}{0\du}}{\pgfpoint{0\du}{0.675000\du}}
\pgfusepath{stroke}
\end{tikzpicture}
\caption{The TN for $\mat{M}(t-1) \times_1 \mat{A}(t-1)$ after contraction of the indices between $\ten{M}^{(k)}$ and $\ten{A}^{(k)} \,(k=1,\ldots,d)$.}
\label{fig:TNMA2}
\end{center}
\end{figure}
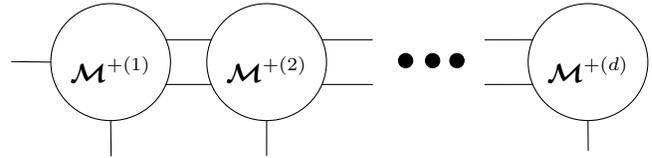
The computation of $\mat{M}^+ = \mat{M}(t-1) \times_1 \mat{A}(t-1)$ in the TN representation is shown in Figure \ref{fig:TNMA}. Note the extra index of the first border core $\ten{M}^{(1)}$, which runs from 1 to $l$. There are many ways to contract the TN of Figure \ref{fig:TNMA} and the total computational complexity depends on the chosen order in which the contractions happen~\cite[p.~126]{TNorus}. For the TN Kalman filter however, the order of contractions is fixed since the resulting TN needs to consist of $d$ TN-cores. This implies that the only allowable contractions are those of $\ten{M}^{(k)}$ with the corresponding $\ten{A}^{(k)}$-core. This contraction is  
\begin{align}
m^{+(k)}_{\alpha_{k-1}\beta_{k-1}j_k\alpha_k\beta_k} &=  \sum_{i_k} m^{(k)}_{\alpha_{k-1}i_k\alpha_k} \; a^{(k)}_{\beta_{k-1}j_ki_k\beta_k},
\label{eq:mplus}
\end{align}
with a computational complexity of $O(r_M^2n^2r_A^2)$. Note that we have permuted the indices after the contraction into the order $\alpha_{k-1}\beta_{k-1}j_k\alpha_k\beta_k$. In this way, \eqref{eq:mplus} is consistent with the contractions as shown in Figure \ref{fig:TNMA2}. The horizontal contractions between consecutive $\ten{M}^{(k)}$ and $\ten{A}^{(k)}$ cores now appear as double contractions between consecutive $\ten{M}^{+(k)}$ cores, which implies that the TN-rank $r_M$ has increased to $r_Mr_A$. This increase of the TN-rank can also be understood from considering both $\alpha_{k-1}\beta_{k-1}$ and $\alpha_k\beta_k$ as multi-indices. The TN-ranks will therefore grow exponentially during the filtering if no measures are taken. Fortunately, it is possible to reduce the TN-ranks without the loss of accuracy by the process of TN-rounding. The rounding procedure for TTs is described in \cite[p.~2301-2305]{ivanTT} and is easily adapted to work for TNs. It involves a right-to-left ``sweep" of QR decompositions over all cores $\ten{M}^{+(k)}$, followed by a left-to-right sweep of singular value decompositions (SVDs). The computational complexity of the whole rounding process is $O(dnr^3)$, where $r$ is the maximal TN-rank. The use of a truncated SVD in the rounding algorithm results in a further reduction of the TN-ranks, at the expense of obtaining only an approximated TN. This can, however, drastically reduce the computation time, as is demonstrated in the experiments in Section \ref{sec:experiments}.

\subsection{TN implementation of $\ten{P}^+= \ten{P}(t-1) \times_1 \mat{A}(t-1) \times_2 \mat{A}(t-1) + \ten{Q}(t-1)$}
\begin{figure}[tb]
\begin{center}
\input{figures/TNPAA.tex}
\caption{The operation $\ten{P}(t-1) \times_1 \mat{A}(t-1) \times_2 \mat{A}(t-1)$ as a TN.}
\label{fig:TNPAA}
\end{center}
\end{figure}
The computation of the predicted covariance matrices consists of two steps. First, there is the contraction $\ten{P}(t-1) \times_1 \mat{A}(t-1) \times_2 \mat{A}(t-1)$ given by
\begin{align}
\nonumber &p^{+(k)}_{\alpha_{k-1}\beta_{k-1}\gamma_{k-1}j_1i_2\alpha_k\beta_k\gamma_k}   \\ 
&=\sum_{i_1,j_2} p^{(k)}_{\alpha_{k-1}i_1j_2\alpha_k} \; a^{(k)}_{\beta_{k-1}i_1j_1\beta_k} \; a^{(k)}_{\gamma_{k-1}i_2j_2\gamma_k},
\label{eq:pplus}
\end{align}
and shown as a TN in Figure \ref{fig:TNPAA}. Again, since we want to obtain the result as a TN of $d$ cores, the order of the contractions is uniquely determined with a computational complexity of $O(r_P^2n^3r_A^4)$. Note that also here the indices of $\ten{P}^{+(k)}$ in \eqref{eq:pplus} are permuted into the order $\alpha_{k-1}\beta_{k-1}\gamma_{k-1}j_1i_2\alpha_k\beta_k\gamma_k$. The TN-ranks $r_P$ are now increased to $r_Pr_A^2$, which means that another rounding step is required. The $d$ cores obtained from the contractions then need to be added to the corresponding $d$ cores of $\ten{Q}(t)$. The addition of tensors in the TT representation is described in \cite[p.~2308]{ivanTT} and its extension is also straightforward for the TN case. It entails the concatenation of the respective cores of the two summands, which means that the TN-ranks of corresponding cores are summed. An additional rounding step is hence required.

\subsection{TN implementation of $\mat{v} = \mat{y}(t) - \mat{M}^+ \times_1 \mat{c}(t)$}
The dominating computational step in the computation of $\mat{v}$ is the contraction $\mat{M}^+ \times_1 \mat{c}(t)$, which in the TN format is
\begin{align*}
\sum_{i_1} m^{+(k)}_{\alpha_{k-1}i_1\alpha_k}\, c^{(k)}_{\beta_{k-1}i_1\beta_{k}}.
\end{align*}
The computational complexity of each core contraction is $O(r_M^2\,r_c^2\,n)$. If the TN-ranks are too large one can at this point apply the rounding procedure. After rounding the whole TN is then contracted, resulting in a vector of length $l$ that can be subtracted from $\mat{y}(t)$.

\subsection{TN implementation of $\mat{s} = \ten{P}^+\times_1 \mat{c}(t) \times_2 \mat{c}(t) + \textrm{diag}(\mat{R}(t))$}
This step involves the contraction of the covariance tensor $\ten{P}^{+}$ with the vector $\mat{c}(t)$ on its first two modes, resulting in a vector of length $l$. In terms of the TN representation, this is achieved through the contraction
\begin{align*}
\sum_{i_1,j_1} p^{+(k)}_{\alpha_{k-1}i_1j_1\alpha_k} \; c^{(k)}_{\beta_{k-1}i_1\beta_k} \; c^{(k)}_{\gamma_{k-1}j_1\gamma_k},
\end{align*}
for each of the $d$ cores with a computational complexity of $O(r_P^2\,r_c^2\,n^2+r_P^2\,r_c^4\,n)$. After rounding and contracting the resulting TN one obtains a row vector of length $l$, which is added with the diagonal entries of the matrix $\mat{R}(t)$.

\subsection{TN implementation of $\mat{K} = \ten{P}^+\times_2 \mat{c}(t) \times_3 \textrm{diag}(\mat{s})^{-1}$}
The contraction $\ten{P}^+\times_2 \mat{c}(t)$ in the TN format is for all cores except the first
\begin{align*}
\sum_{j_1} p^{+(k)}_{\alpha_{k-1}i_1j_1\alpha_k} c^{(k)}_{\beta_{k-1}j_1\beta_k}
\end{align*}
and has a computational complexity of $O(r_P^2\,r_c^2\,n^2)$. If we denote the $l \times l$ matrix $\textrm{diag}(\mat{s})^{-1}$ by $\mat{S}$, then the scaling operation $\times_3\, \mat{S}$ corresponds with a contraction on only the first core $\ten{P}^{(1)}$ such that all contractions on it are
\begin{align*}
\sum_{j_1,\alpha_{0}} p^{+(1)}_{\alpha_0i_1j_1\alpha_1}\; c^{(1)}_{\beta_{0}j_1\beta_1} \; s_{\alpha_0\gamma_{1}}.
\end{align*}
An additional rounding step can be applied to reduce the TN-ranks of the resulting Kalman gain TN.

\subsection{TN implementation of $\mat{M}(t) = \mat{M}^+ + \mat{K}\times_2 \textrm{diag}(\mat{v})$}
The updated mean matrix $\mat{M}(t)$ is obtained from adding $\mat{M}^+$ with a scaled version of the Kalman gain. The scaling of the Kalman gain in the TN format is performed in an identical way as the scaling $\times_3\, \mat{S}$ from the previous step with a computational complexity of $O(l^2nr_K)$. All TN-cores $\ten{K}^{(k)}$ are then concatenated with the cores $\ten{M}^{+(k)}$ to obtain the TN representation of $\mat{M}(t)$. The concatenation adds all corresponding TN-ranks together, which implies that a rounding step is required.

\subsection{TN implementation of $\ten{P}(t) = \ten{P}^+ - (\mat{K}\,  \square \, \mat{K})\times_3 \textrm{diag}(\mat{s})$}
In this step the column-wise outer product of the Kalman gain matrix $\mat{K}$ with itself needs to be computed in the TN format. The following lemma describes how this can be done.
\begin{lemma}
Let $\mat{K}_1$ be the $l\times nr_K$ matrix obtained from reshaping the $\ten{K}^{(1)}$ core, then the first TN core of $\mat{K}\,  \square \, \mat{K}$ is obtained by the following procedure:
\begin{enumerate}
\item compute the matrix $\mat{K}_{11}:= \mat{K}_1^T \odot \mat{K}_1^T$,
\item reshape $\mat{K}_{11}$ into a $n \times r_K \times  n \times r_K \times l$ tensor $\ten{K}_{11}$,
\item permute $\ten{K}_{11}$ into a $l \times n  \times n \times r_K \times r_K$ tensor $\tilde{\ten{K}}_{11}$,
\item reshape $\tilde{\ten{K}}_{11}$ into the desired $l \times n^2  \times  r_K^2$ tensor core.
\end{enumerate}
The remaining $d-1$ cores are $\ten{K}^{(k)} \otimes \ten{K}^{(k)} \, (k=2,\ldots,d)$.
\end{lemma}
\begin{proof}
If we fix the row indices $i_1,\ldots,i_d$ and column indices  $i_{d+1},\ldots,i_{2d}$ of $(\mat{K}\,  \square \, \mat{K})$ and fix the index $j$ for its third mode, then we have that
\begin{align*}
(\mat{K}\,  \square \, \mat{K})_{[i_1\cdots i_d][i_{d+1}\cdots i_{2d}]j} &= k_{[i_1\cdots i_d]j}\, k_{[i_{d+1}\cdots i_{2d}]j}.
\end{align*}
In the TN representation fixing these indices results in
\begin{align}
(\mat{k}^{(1)}_{ji_1} \otimes \mat{k}^{(1)}_{ji_{d+1}}) \, (\mat{K}^{(2)}_{i_2} \otimes \mat{K}^{(2)}_{i_{d+2}}) \cdots (\mat{k}^{(d)}_{i_d} \otimes \mat{k}^{(d)}_{i_{2d}}),
\label{eqn:Kproof}
\end{align}
where $\mat{k}^{(1)}_{ji_1}$ denotes the row vector obtained from fixing the indices of the first two modes of $\ten{K}^{(1)}$. Similarly, $\mat{K}^{(2)}_{i_2}$ denotes the matrix obtained from fixing the second index of $\ten{K}^{(2)}$. Fixing the second index of $\ten{K}^{(d)}$ results in a column vector $\mat{k}^{(d)}_{i_d}$. The Kronecker products of these vectors and matrices in \eqref{eqn:Kproof} can be rewritten using the mixed product property of the Kronecker product into
\begin{align*}
&({\mat{k}}^{(1)}_{ji_1}\, \mat{K}^{(2)}_{i_2} \cdots \mat{k}^{(d)}_{i_d})  \otimes (\mat{k}^{(1)}_{ji_{d+1}} \, \mat{K}^{(2)}_{i_{d+2}} \cdots  \mat{K}^{(d)}_{i_{2d}}),\\
&= k_{[i_1\cdots i_d]j}\, k_{[i_{d+1}\cdots i_{2d}]j},
\end{align*}
which concludes the proof.
\end{proof}
The tensor Kronecker product sets the computational complexity of computing each core to $r_K^4n^2$. The scaling of $\mat{K}\,  \square \, \mat{K}$ with $\mat{S}$ is exactly the same as the scaling of the Kalman gain, which is via a contraction of the first core. After the result of the scaling is added to $\ten{P}^{+}$, an additional rounding step can be applied to reduce the TN-ranks. A complete overview of the computational complexity of each of the TN Kalman filter steps computed in both the standard and TN ways is given in Table \ref{tab:comp}. From this table one can see that using TNs transforms the computational complexity from exponential to linear in the degree $d$. The importance of the rounding procedure also becomes clear since the computational complexity is polynomial in the TN-ranks $r_M,r_P,r_K,r_A,r_c$. Having derived the TN Kalman filter, we now move on to discuss its application in the recursive identification of MIMO Volterra systems.
\begin{table}[tb]
\begin{center}
\caption{Computational complexity of the standard and TN approach.}
\label{tab:comp}	
\begin{tabular}{@{}lcc@{}}
		 &  \multicolumn{2}{c}{Computational complexity (flops)} \\
\midrule
 Computation		& Standard		& TN \\ \midrule
$\mat{M}^{+}$ 		& $O(ln^{2d})$ 		& $O(d\,r_M^2\,r_A^2\,n^{2})$  \\
$\ten{P}^{+}$ 		& $O(ln^{3d})$ 		& $O(d\,r_P^2\,r_A^4\,n^{3})$\\
$\mat{v}$ 			& $O(ln^d)$ 		& $O(d\,r_M^2\,r_c^2\,n)$ \\
$\mat{s}$ 			& $O(ln^{2d})$ 		& $O(d(r_P^2\,r_c^2\,n^2+r_P^2\,r_c^4\,n))$\\
$\mat{K}$			& $O(ln^{2d})$		& $O(d\,r_P^2\,r_c^2\,n^2)$\\
$\mat{M}(t)$ 		& $O(ln^{d})$		& $O(l^2\,r_K\,n)$\\
$\ten{P}(t)$ 		& $O(ln^{2d})$	& $O(d\,r_K^4\,n^2)$
\end{tabular}
\end{center}
\end{table}

\section{Recursive MIMO Volterra system identification}
\label{sec:VolterraTensor}
The first required step in applying a TN Kalman filter to the recursive identification of the Volterra kernel coefficients is writing the Volterra system as a TN. What follows is a brief review of the TN representation of a discrete-time MIMO Volterra system presented in \cite{MVMALS}. A discrete-time $p$-input $l$-output Volterra system of degree $d$ and memory $M$ is described by
\begin{align}
\mat{y}(t) &:= \begin{pmatrix} y_1(t) & y_2(t)& \cdots& y_l(t) \end{pmatrix} = (\mat{u}_t\kpr{d})^T \mat{V} \; ,
\label{eq:defMIMOVolterra}
\end{align}
where the vector
\begin{align*}
\mat{u}_t := \begin{pmatrix}1 & u_1(t) & u_2(t) & \cdots & u_p(t-M+1) \end{pmatrix}^T \in \mathbb{R}^{pM+1}
\end{align*}
contains all $p$ input values at times $t$ down to $t-M+1$ and $\mat{u}_t\kpr{d}$ denotes its $d$-times repeated Kronecker product. Each column of the $(pM+1)^d \times l$ matrix $\mat{V}$ contains all coefficients from the Volterra kernels of degree $0$ up to degree $d$ for a particular output. The repeated Kronecker product structure of $\mat{u}_t$ can be exploited to rewrite \eqref{eq:defMIMOVolterra} as
\begin{align}
\mat{y}(t) &= \ten{V} \times_2 \mat{u}_t^T \times_3 \mat{u}_t^T \cdots \times_{d+1} \mat{u}_t^T,
\label{eq:MIMOVolterraTensor}
\end{align}
where $\ten{V}$ is the tensor obtained from reshaping $\mat{V}^T$ in \eqref{eq:defMIMOVolterra} into an $l \times (pM+1) \times \cdots \times (pM+1)$ tensor. The MIMO Volterra tensor $\ten{V}$ consists of $l\,(pM+1)^d$ entries, which quickly becomes infeasible to store even for moderately large values of $p,M$ and $d$. We therefore represent the Volterra tensor $\ten{V}$ as a TN $\ten{V}^{(1)},\ldots,\ten{V}^{(d)}$, where the first TN-core $\ten{V}^{(1)}$ has dimensions $l \times (pM+1) \times r_1$. The other TN-cores have sizes $r_{i-1} \times (pM+1) \times r_i$, with $r_d:=1$ for the last core. The TN reduces to a TT for the MISO case $(l=1)$. This change in representation reduces the storage requirement from $l\,(pM+1)^d$ to $O((d-1)(pM+1)r^2+(pM+1)lr)$. In \cite{MVMALS} it is also described how the $l$ output samples at time $t$ can be all be computed at once in the TN format as
\begin{align}
\mat{y}(t) &= (\ten{V}^{(1)} \times_2 \mat{u}_t^T) \, (\ten{V}^{(2)} \times_2 \mat{u}_t^T)  \cdots (\ten{V}^{(d)} \times_2 \mat{u}_t^T),
\label{eq:simTT}
\end{align}
with a computational complexity of $O(d(pM+1)r+dr^3)$. The corresponding TN diagram is shown in Figure \ref{fig:MIMOVolterraTN} and is essentially the TN representation of a discrete-time MIMO Volterra system. The connecting branches between the $\mat{u}_t$ nodes in Figure \ref{fig:MIMOVolterraTN} represent unit TN-ranks and correspond with $d$ outer products of the $\mat{u}_t$ vector with itself, resulting in a $d$-way tensor $\ten{U}$. An alternative way of writing \eqref{eq:simTT} is therefore
\begin{align*}
y_{i_1}(t) &= \sum_{i_2,i_3,\cdots ,i_{d+1}} \, v_{i_1i_2i_3\cdots i_{d+1}} \; u_{i_2i_3\cdots i_{d+1}}.
\end{align*}
where the summation runs over all repeated indices $i_2,i_3,i_4,i_5,\ldots,i_{2d-1},i_{2d}$. The only index that remains after the contraction is $1 \leq i_1 \leq l$, which runs over all entries of the $\mat{y}(t)$ vector. 
\begin{figure}[tb]
\begin{center}
\input{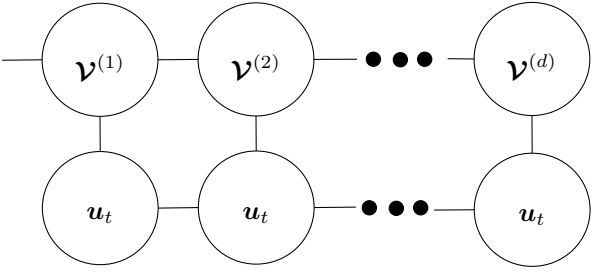}
\caption{A TN representation of a discrete-time MIMO Volterra system.}
\label{fig:MIMOVolterraTN}
\end{center}
\end{figure}
Next, we reformulate the MIMO Volterra system as a linear state space system such that the system identification problem can be recursively solved using a Kalman filter in the TN format. In \cite{Barner2006}, a time-varying discrete-time SISO Volterra system is described as a linear state space model, which can be extended to the MIMO case as
\begin{align}
\nonumber \mat{V}(t+1) &= \mat{A}(t)\, \mat{V}(t) + \mat{W}(t),\\ 
\mat{y}(t) &= (\mat{u}_t\kpr{d})^T \mat{V}(t)  + \mat{r}(t).
\end{align}
The row vector $\mat{c}$ in \eqref{eqn:linearstatespace} is for this particular case replaced with $(\mat{u}_t\kpr{d})^T$. The Volterra system is time-invariant when $\mat{A}(t) = \mat{I}$ and $\mat{W}(t)=\mat{O}$. If $\mat{W}(t)$ is a nonzero matrix instead, then the Volterra kernel coefficients follow a random walk. As mentioned earlier, the repeated Kronecker product structure of $\mat{u}_t$ gives rise to a tensor $\ten{U}$ with the following TN representation.
\begin{lemma}
All $d$ TN-cores $\ten{U}^{(1)},\ldots,\ten{U}^{(d)}$ of $\ten{U}$ have unit TN-ranks with
\begin{align*}
\ten{U}^{(k)}(1,:,1) &= \mat{u}_t \quad (k=1,\ldots,d).
\end{align*}
\label{lemma:ut}
\end{lemma}
The TN for $\ten{U}$ therefore consists of $\mat{u}_t$ repeated $d$ times, with a total storage cost of $O(pM+1)$ since we only need to store the $\mat{u}_t$ vector once. The fact that all its TN-ranks are equal to 1 for any time instance $t$ is very fortunate as this significantly reduces both the storage cost and computational complexity. Indeed, one can set $r_c:=1$ and $n:=(pM+1)$ in Tables \ref{tab:storage} and \ref{tab:comp} to obtain the storage cost and computational complexity for each step of the recursive system identification with the TN Kalman filter.

\section{Application}
\label{sec:experiments}
We demonstrate the effectiveness of the proposed TN Kalman filter through its application to the system identification problem of MIMO Volterra systems. All computations were performed in Matlab on an Intel i5 quad-core processor running at 3.3 GHz with 16 GB RAM. All TN computations are implemented in Matlab/Octave and our algorithms are freely available from \url{https://github.com/kbatseli/TNKalman}.

\subsection{Degree-4 SISO time-invariant Volterra system}
As a first experiment we compare the use of a standard Kalman filter with the TN Kalman filter to confirm its correctness. We also demonstrate the effect of using a truncated SVD in the rounding process on the runtime of the filter. The applicability of the standard Kalman filter is limited by the length of the state vector. We therefore consider the following time-invariant SISO Volterra system
\begin{align*}
\nonumber \mat{v}(t+1)\kpr{4} &= \mat{v}(t)\kpr{4} ,\\ 
y(t) &= (\mat{u}_t\kpr{4})^T \,  \mat{v}(t)\kpr{4} + r(t),
\end{align*}
where $\mat{v}(0) \in \mathbb{R}^5$ and each entry is drawn from a standard normal distribution. The measurement noise is described by $r(t) \sim \textrm{N}(0,10^{-2})$. The memory $M$ is set to 4 and all input samples $u(t)$ are also drawn from a standard normal distribution. This linear state space system corresponds with a degree-4 Volterra system. Since the number of outputs $l$ is 1, we have that the state $\mat{V}(t)$ is reduced to a vector $\mat{v}(t)$. The absence of process noise $\mat{w}(t)$, together with the fact that $\mat{A}(t):=\mat{I}$ implies that the Volterra system is time-invariant. The state vector $\mat{v}$ of Volterra kernel coefficients has a length of $(4+1)^4=625$ and the standard Kalman filter is therefore well-suited for the recursive estimation of the kernel coefficients. The mean vector $\mat{m}(0)$ is initialized as the zero vector and the covariance matrix $\mat{P}(0):= 1000\,\mat{I}_{625}$. A standard Kalman filter is then run for 1000 iterations. The median of the runtime over the 1000 iterations of the Kalman filter is 0.0018 seconds. Graph A in Figure \ref{fig:ex1} shows the relative error
\begin{align*}
\frac{||\mat{v}(0)\kpr{4} - \mat{m}(t)||_2}{|| \mat{v}(0)\kpr{4}||_2}
\end{align*}
as a function of the iterations for the standard Kalman filter. When the TN Kalman filter is initialized with the same mean $\mat{m}(0)$ and covariance matrix $\mat{P}(0)$ and no truncated SVD is used in the rounding step, then the exact same graph A is reproduced, which confirms the correctness of the TN Kalman filter. The TN-ranks for the TN representing $\mat{m}$ and $\mat{P}$ converge to $5,15,5$ and $25,226,25$, respectively. The relatively large TN-rank of 226 in the TN of $\mat{P}$ sets the median computation time to $0.1456$ seconds, which is about 80 times slower than a standard Kalman filter. Graphs B,C and D in Figure \ref{fig:ex1} show the relative error for the TN Kalman filter where tolerances of $0.9,0.5$ and $0.1$ are used to truncate the SVD in the rounding step, resulting in a median runtime of 0.0026 seconds per iteration, which is the same order of magnitude as the standard Kalman filter. These tolerances reduce all TN-ranks of $\mat{m}$ and $\mat{P}$ to 1 and come at the cost of slower convergence as demonstrated in Figure \ref{fig:ex1}. Setting the rounding parameter higher than $10^{-1}$ results in this example in a convergence of the TN Kalman filter to estimated Volterra coefficients with a higher relative error.

\begin{figure}[tb]
\begin{center}
\includegraphics[width=.5\textwidth]{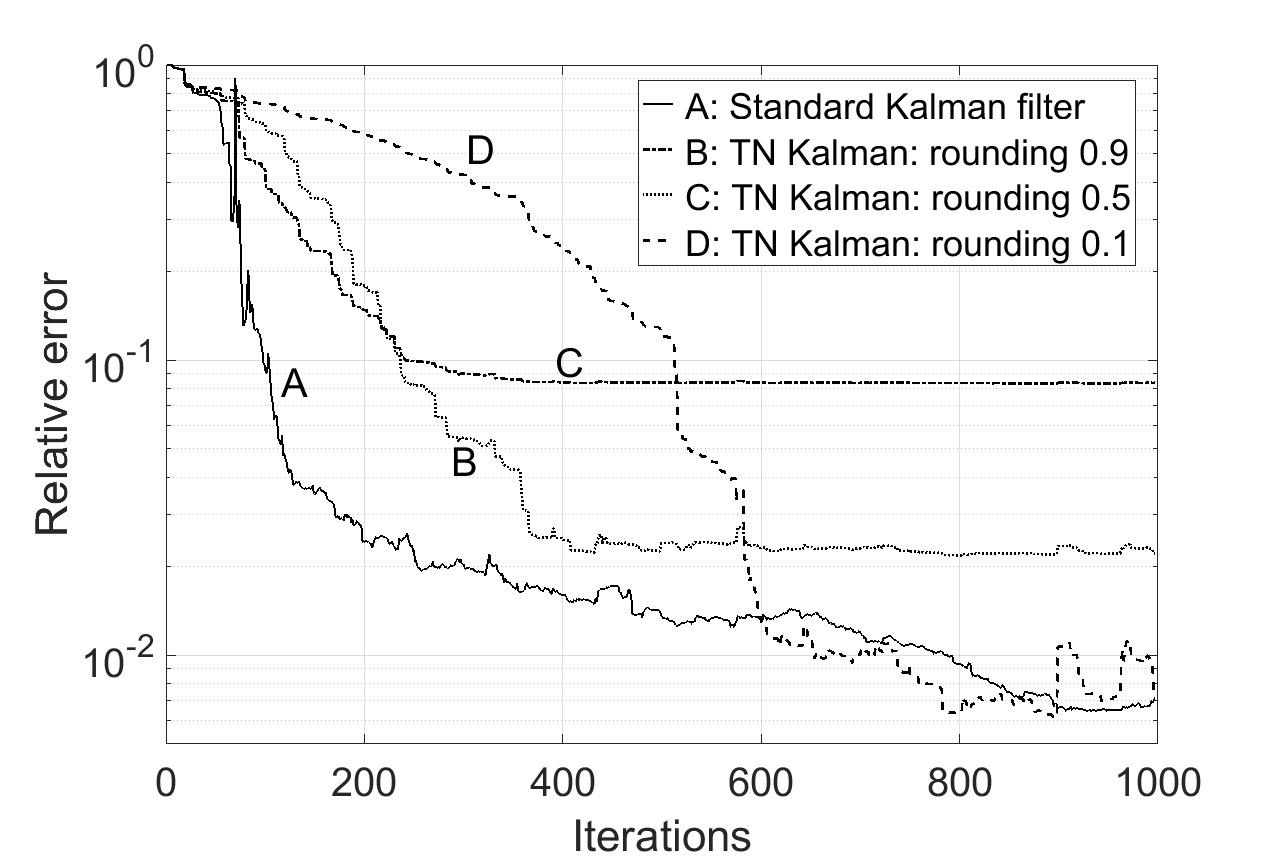}
\caption{Relative errors from both standard and TN Kalman filters for different values of the rounding parameter.}
\label{fig:ex1}
\end{center}
\end{figure}

\subsection{Double balanced mixer}
In this example we consider a double balanced mixer used for upconversion. The output radio-frequency (RF) signal is determined by a 100Hz sine low-frequency (LO) signal and a 300Hz square-wave intermediate-frequency (IF) signal. A phase difference of $\pi/8$ is present between the LO and IF signals. All time series were sampled at 5 kHz for a total of 6000 samples. The output signal is corrupted with Gaussian noise such that three different outputs with respective signal-to-noise ratios (SNRs) of 12 dB, 17 dB and 26 dB are obtained. A TN Kalman filter is then used to estimate a two-input one-output Volterra system with $d=7,M=10$ by filtering the first 5900 samples for the three noisy outputs separately. The state vector containing the Volterra kernel coefficients consists of $21^7\approx 1.801\times 10^{9}$ entries, which is well beyond the reach of a standard Kalman filter.

The initial mean vector is initialized as the zero vector and the initial variance of each of the coefficients is set to 1000. A rounding parameter of $10^{-1}$ is used to keep the TN-ranks small. The maximal TN-ranks of $\ten{M}$ when filtering the output with SNRs of 12 db, 17 dB and 26 dB are $11,13$ and $14$, respectively. All TN-ranks of $\ten{P}$ are equal to one. The median runtime for one TN Kalman filter step is $0.0068$ seconds and the total runtime to filter 5900 samples is about 40 seconds. The obtained mean vectors in the TN format were then used to simulate the remaining 100 samples. The simulated outputs are shown together with the reference output, which is not corrupted by noise, in Figure \ref{fig:ex2}. Figure \ref{fig:ex2} demonstrates that a higher SNR results in better performance of the TN Kalman filter. For the outputs with a SNR of 12 dB, 17 dB and 26 dB, the root-mean-square errors (RMSE) of the simulated outputs are $0.1778,0.097$ and $0.034$, respectively. The reference output, the 12 dB output used in the identification and the simulated output are shown in Figure \ref{fig:ex2b}, where one can see that the simulated output follows the reference more closely.

\begin{figure}[tb]
\begin{center}
\includegraphics[width=.5\textwidth]{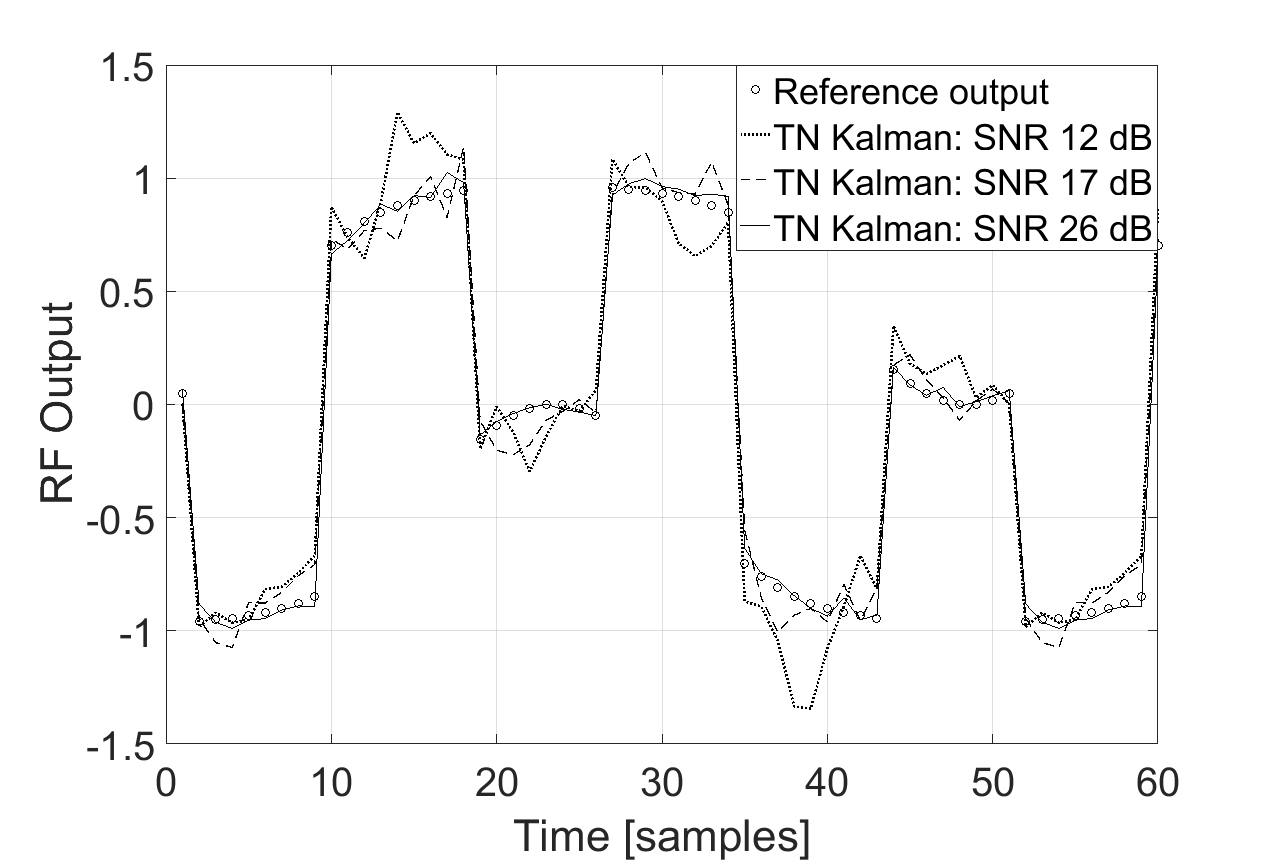}
\caption{Reference and simulated output from the TN Kalman filter for three different SNRs.}
\label{fig:ex2}
\end{center}
\end{figure}
\begin{figure}[tb]
\begin{center}
\includegraphics[width=.5\textwidth]{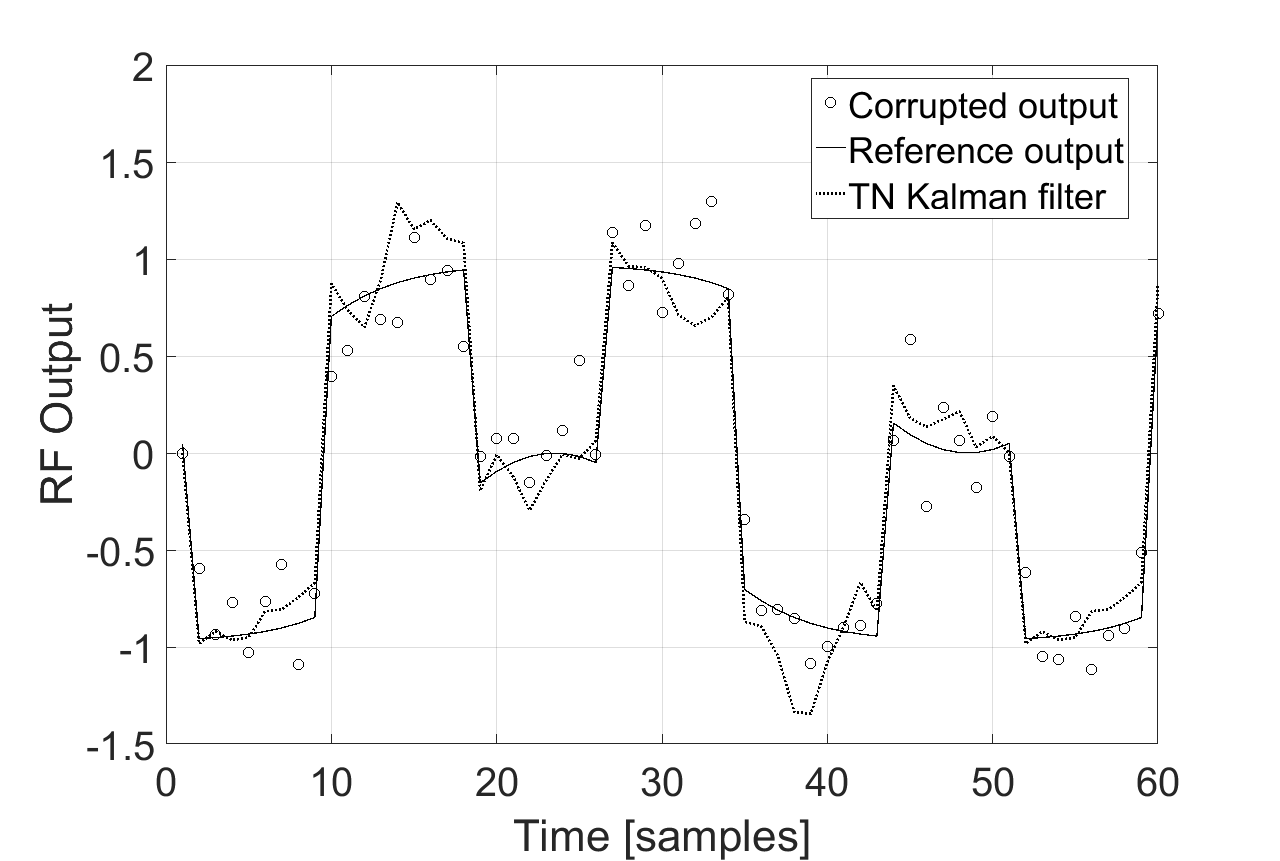}
\caption{Reference, corrupted (SNR of 12 dB) and simulated output.}
\label{fig:ex2b}
\end{center}
\end{figure}

\section{Conclusions}
This article presented a Tensor Network Kalman filter with an application in the recursive identification of high-order discrete-time nonlinear MIMO Volterra systems. Tensor Network theory enables the estimation of exponentially large state vectors, without ever needing to explicitly construct them. This allows the Kalman filter to be applied to previously prohibitive problem scales. The correctness of the Tensor Network Kalman filter and its efficient computation were demonstrated via numerical experiments. Future improvements are the implementation of a square-root TN Kalman filter, together with the extension to matrix outputs. 

\bibliographystyle{plain}        
\bibliography{references}           



\end{document}